\documentclass[12pt,reqno]{amsart}
\usepackage{amsmath,amssymb,bbm,booktabs,esint,times,url}

\usepackage{tikz}
\usepackage[bookmarks=false,unicode,colorlinks,urlcolor=red!70!black,citecolor=red!70!black,linkcolor=blue,
]{hyperref}
\usepackage[all]{hypcap}
\usepackage{aliascnt}
\usetikzlibrary{calc}

\usepackage[margin=1.25in]{geometry}

\newtheorem*{remark}{Remark}

\newtheorem{theorem}{Theorem}[section]

\newaliascnt{lemma}{theorem}
\newtheorem{lemma}[lemma]{Lemma}
\aliascntresetthe{lemma}

\newaliascnt{proposition}{theorem}

\aliascntresetthe{proposition}

\newaliascnt{corollary}{theorem}
\newtheorem{corollary}[corollary]{Corollary}
\aliascntresetthe{corollary}

\newaliascnt{conjecture}{theorem}

\aliascntresetthe{conjecture}

\makeatletter
\def\tagform@#1{\maketag@@@{\ignorespaces#1\unskip\@@italiccorr}}
\let\orgtheequation\theequation
\def\theequation{(\orgtheequation)}
\makeatother
\def\equationautorefname~{}

\renewcommand{\C}{{\mathbb C}}
\newcommand{\D}{{\mathbb D}}
\newcommand{\e}{\varepsilon}

\newcommand{\psiplus}{{\psi_+}}
\newcommand{\psiminus}{{\psi_-}}
\newcommand{\Ray}{\operatorname{Ray}}
\newcommand{\R}{{\mathbb R}}
\newcommand{\Real}{\operatorname{Re}}
\newcommand{\spec}{\operatorname{spec}}
\newcommand{\Z}{{\mathbb Z}}

\begin{document}

\title[]{Magnetic spectral bounds on starlike plane domains}

\author[]{R. S. Laugesen and B. A. Siudeja}
\address{Department of Mathematics, Univ.\ of Illinois, Urbana,
IL 61801, U.S.A.}
\email{Laugesen\@@illinois.edu}
\address{Department of Mathematics, Univ.\ of Oregon, Eugene,
OR 97403, U.S.A.}
\email{Siudeja\@@uoregon.edu}
\date{\today}


\begin{abstract}
We develop sharp upper bounds for energy levels of the magnetic Laplacian on starlike plane domains, under either Dirichlet or Neumann boundary conditions and assuming a constant magnetic field in the transverse direction. Our main result says that $\sum_{j=1}^n \Phi \big( \lambda_j A/G \big)$ is maximal for a disk 
whenever $\Phi$ is concave increasing, $n \geq 1$, the domain has area $A$, and $\lambda_j$ is the $j$-th Dirichlet eigenvalue of the magnetic Laplacian $\big( i\nabla+ \frac{\beta}{2A}(-x_2,x_1) \big)^2$. Here the flux $\beta$ is constant, and the scale invariant factor $G$ penalizes deviations from roundness, meaning $G \geq 1$ for all domains and $G=1$ for disks. 
\end{abstract}

\maketitle

\vspace*{-12pt}

\section{\bf Introduction}
\label{sec:intro}

\subsection*{Overview}
The energy levels of a charged quantum particle in a two dimensional region are difficult to understand analytically. We aim for insight into the behavior of these energy levels by proving that they are maximal for a certain disk whose radius is computed from the boundary shape of the original confinement region. 

Specifically, we develop sharp upper bounds for energy levels of the magnetic Laplacian on starlike plane domains, under either Dirichlet or Neumann boundary conditions, assuming a constant magnetic field in the transverse direction. The spectral functionals we consider include the ground state energy, sums and products of energy levels, the spectral zeta function, and the partition function. 

For the special case of the ground state energy, our upper bound complements a lower bound of Faber--Krahn type due to Erd\"{o}s \cite{E96}, which says that the first eigenvalue $\lambda_1$ of the magnetic Dirichlet Laplacian is minimal for the disk of the same area. Combining these upper and lower bounds gives a pair of inequalities:
\[
1 \leq \frac{\lambda_1(\Omega)}{\lambda_1(\Omega^*)} \leq G(\Omega)
\]
when $\Omega \subset \R^2$ is a starlike plane domain, $\Omega^*$ is the disk of the same area, and the computable geometric factor $G(\Omega)$ measures how far the domain is from being circular (with $G=1$ for a disk; see the definition in the next section). Note that the upper estimate in this paper requires starlikeness of the domain, whereas the lower estimate due to Erd\"{o}s holds whenever the domain is merely bounded. The upper estimate has the advantage of applying also under Neumann boundary conditions. 

Our results apply to a huge class of spectral functionals beyond the ground state energy. \autoref{Deigen} shows that
\[
\sum_{j=1}^n \Phi \Big( \frac{\lambda_j(\Omega)}{G(\Omega)} \Big) \leq \sum_{j=1}^n \Phi \big( \lambda_j (\Omega^*) \big)
\]
whenever $n \geq 1$ and $\Phi : \R_+ \to \R$ is concave and increasing.

\subsection*{Formulating the problem}
Let us begin with some physical background, and then formulate the results precisely.  Impose a vertical magnetic field of constant strength through a cylinder $\Omega \times \R$, and let the magnetic Schr\"{o}dinger operator act upon a charged, spinless quantum particle that is confined to the cylinder. The particle moves freely in the vertical direction, and so its wavefunction can be written in separated form as a plane wave in the vertical direction multiplied by an eigenfunction of the magnetic Laplacian in the horizontal directions. The energy levels of the horizontal motion are the objects of our study. 

To state the problem mathematically, consider a bounded plane domain $\Omega$ with area $A$, and fix a real number $\beta$. The magnetic Laplacian on $\Omega$ is the symmetric operator 
\[
(i \nabla + F)^2
\]
where the vector potential 
\[
F(x)=\frac{\beta}{2A}(-x_2,x_1)
\]
is chosen to generate a transverse magnetic field $\nabla \times F = (0,0,\beta/A)$ of strength $B=\beta/A$, as illustrated in \autoref{fi:magfig}. The constant $\beta$ represents the magnetic flux through the domain. 

\begin{figure}[t]
  \begin{center}
    \vspace{-0.6cm}
\begin{tikzpicture}[scale=1]
  \draw (0,0) -- (3,0) -- node[below]{$\Omega$} (4,1) -- (1,1) -- cycle;
  \draw (3/4,3/2) -- ++(0,-2/6) coordinate (a);
  \draw[<-] (a) -- ++(0,-5/6) coordinate (a);
  \draw[dotted] (a) -- (3/4,0) coordinate (b);
  \draw (b) -- ++(0,-1/3) coordinate (b);
  \draw[<-] (b) -- ++(0,-1/3);
  \draw (2,5/3) node[left]{$\mathbf B$}-- ++(0,-2/6) coordinate (a);
  \draw[<-] (a) -- ++(0,-5/6) coordinate (a);
  \draw[dotted] (a) -- (2,0) coordinate (b);
  \draw (b) -- ++(0,-1/4) coordinate (b);
  \draw[<-] (b) -- ++(0,-1/4);
  \draw (2.75,3/2) -- ++(0,-2/6) coordinate (a);
  \draw[<-] (a) -- ++(0,-5/6) coordinate (a);
  \draw[dotted] (a) -- (2.75,0) coordinate (b);
  \draw (b) -- ++(0,-1/3) coordinate (b);
  \draw[<-] (b) -- ++(0,-1/3);
\end{tikzpicture}
  \end{center}
\caption{\label{fi:magfig} A plane domain subjected to a transverse magnetic field.}
\end{figure}
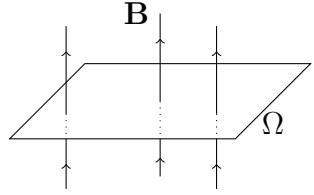

Obviously the magnetic Laplacian reduces to the usual Laplacian in the absence of a magnetic field, that is, when $\beta=0$ and $F \equiv 0$.

The magnetic Laplacian has discrete spectrum, assuming Dirichlet boundary conditions on $\partial \Omega$, with eigenvalues $\{ \lambda_j \}$ satisfying
\[
0 < \lambda_1 \leq \lambda_2 \leq \lambda_3 \leq \dots
\]
We denote by $\{ u_j \}$ a corresponding sequence of $L^2$-orthonormal eigenfunctions, with
\begin{equation} \label{eq:specD}
\begin{cases}
(i \nabla + F)^2 u_j = \lambda_j u_j \;\; \text{in $\Omega$,} \\
\hfill u_j = 0 \;\; \text{on $\partial \Omega$.}
\end{cases}
\end{equation}
The normalized eigenvalues $\lambda_j A$ are invariant under dilation of the domain, as one can check straightforwardly using that the field strength $\beta/A$ scales inversely with the area. (For a leisurely treatment of this and other invariance properties of the spectrum of the magnetic Laplacian, see \cite[Section 2 and Appendix A]{LLR12}.) 

Assume $\Omega$ is a \emph{Lipschitz-starlike plane domain}, by which we mean 
\[
\Omega = \{ re^{i\theta} : 0 \leq r < R(\theta) \}
\]
where the radius function $R(\cdot)$ is positive, $2\pi$-periodic, and Lipschitz continuous. Define two scale-invariant geometric factors in terms of the radius function, by 
\begin{align*}
G_0 & = 1 + \frac{1}{2\pi} \int_0^{2\pi} (\log R)^\prime(\theta)^2 \, d\theta , \\
G_1 & = \frac{\frac{1}{2\pi} \int_0^{2\pi} R(\theta)^4 \, d\theta}{\big( \frac{1}{2\pi} \int_0^{2\pi} R(\theta)^2 \, d\theta \big)^2 } = \frac{2\pi I_\text{origin}}{A^2} , 
\end{align*}
where $I_\text{origin}=\int_\Omega |x|^2 \, dx$ is the polar moment of inertia of $\Omega$ about the origin. Obviously
\[
G_0 \geq 1  \qquad \text{and} \qquad G_1 \geq 1
\]
with equality if and only if the domain is a disk centered at the origin $(R \equiv \text{const}$).

\begin{figure}
\begin{center}
\begin{tikzpicture}[scale=1.3]
  \draw (0,0) +(0:1) coordinate (g) -- +(45:2) coordinate (a) -- +(105:0.5) coordinate (b) -- +(145:2) coordinate (e) -- +(180:1) coordinate (f) -- +(240:1.6) coordinate (c) -- +(260:0.5) coordinate (d) -- +(320:1.5) -- +(0:1);

  \draw (0,0) -- ($(a)!0.5!(b)$) coordinate (x) 
  node [right=-2pt,pos=0.6] {\scriptsize $R(\theta)$};




  \fill (0,0) circle (0.03) node [below right=-2pt] {\small $0$};
  
  \draw[dashed] (0,0) -- (1,0);
  \clip (0,0) -- (1,0) -- (x)  -- cycle;
  \draw (0,0) circle (0.5);
  \draw (0.1,0) node [above right] {\scriptsize $\theta$};
\end{tikzpicture}
\end{center}
\caption{A starlike domain with radius function $R(\theta)$.}
\label{fig:starlike}
\end{figure}
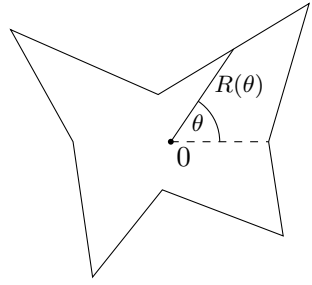

Take the maximum of the two geometric factors, and call it $G$: 
\[
G = \max \{ G_0,G_1 \} \geq 1. 
\]
We interpret $G$ as measuring the deviation of the domain from roundness. Deviation can occur in two ways: an oscillatory boundary would make $R^\prime$ large and hence $G_0$ large, whereas an elongated boundary (such as an eccentric ellipse) would force $R^4$ to vary more than $R^2$ and hence would make $G_1$ large. Calculations are generally required in order to determine which of $G_0$ or $G_1$ is larger (see \cite[Section 10]{LS13}).

\section{\bf Main results}
\label{sec:results}

\subsection*{Dirichlet boundary conditions}
Our main result says that the disk maximizes eigenvalues of the magnetic Laplacian under suitable geometric scaling normalized by area and $G$. 
\begin{theorem}[Dirichlet magnetic Laplacian] \label{Deigen}
Suppose $\Omega = \{ re^{i\theta} : 0 \leq r < R(\theta) \}$ is a Lipschitz-starlike plane domain. 
Fix $\beta \in \R$ and $n \geq 1$. 

Then each of the following scale invariant eigenvalue functionals achieves its maximum value when the domain $\Omega$ is a centered disk:
\[
\lambda_1 A/G, \qquad
(\lambda_1^s + \cdots + \lambda_n^s)^{1/s} \, A/G, \qquad
\sqrt[n]{\lambda_1 \lambda_2 \cdots \lambda_n} \, A/G,
\]
for each exponent $0 < s \leq 1$. Further, if $\Phi : \R_+ \to \R$ is concave and increasing then $\sum_{j=1}^n \Phi(\lambda_j A/G)$ is maximal when $\Omega$ is a disk centered at the origin. 

Hence each partial sum of the spectral zeta function and of the trace of the heat kernel is minimal when $\Omega$ is a centered disk. That is, the functionals
\[
\sum_{j=1}^n (\lambda_j A/G)^s \qquad \text{and} \qquad
\sum_{j=1}^n \exp(-\lambda_j At/G)
\]
attain their smallest value when $\Omega$ is a centered disk, for each $s<0<t$.

Equality statement: if $\lambda_1 A /G \big|_\Omega = \lambda_1 A /G \big|_\D$ then $\Omega$ is a centered disk.
\end{theorem}
The proof appears in \autoref{Deigen_proof}. 

Eigenvalues of the magnetic Laplacian on a disk are extremal, in the theorem. The proof does not need formulas for them, though they  can be computed in terms of zeros of certain Kummer functions --- see the detailed treatment by Son \cite{S14}, which includes informative plots of the eigenvalues as functions of the flux $\beta$.

The theorem can be strengthened by replacing the maximum of $G_0$ and $G_1$, which we call $G$, with certain convex combinations of $G_0$ and $G_1$: see our discussion in the case of zero magnetic field \cite[Section 9]{LS13} . Further improvements can be made by choosing a ``good''  location for the origin, so as to reduce the values of $G_0$ and $G_1$ \cite[Section 10]{LS13}.

We prove the theorem by transforming $\Omega$ into a disk while controlling angular information in the Rayleigh quotient of the magnetic Laplacian. Our transformation is linear on rays and has constant Jacobian, as  indicated in \autoref{fig:transplant}. Note that wherever the transformation stretches radially it must compress angularly, in order to preserve area; this constant Jacobian condition guarantees that when we transplant orthogonal eigenfunctions from the disk we will obtain orthogonal trial functions on $\Omega$. 

One cannot know which orientation yields the smallest value for the Rayleigh quotient of our trial function in $\Omega$, and in any case the optimal orientation will typically differ for each index $j$. We aim instead for the average case: we consider all possible orientations of the trial function in $\Omega$ by employing the arbitrary rotation $U$ of the disk in \autoref{fig:transplant}.

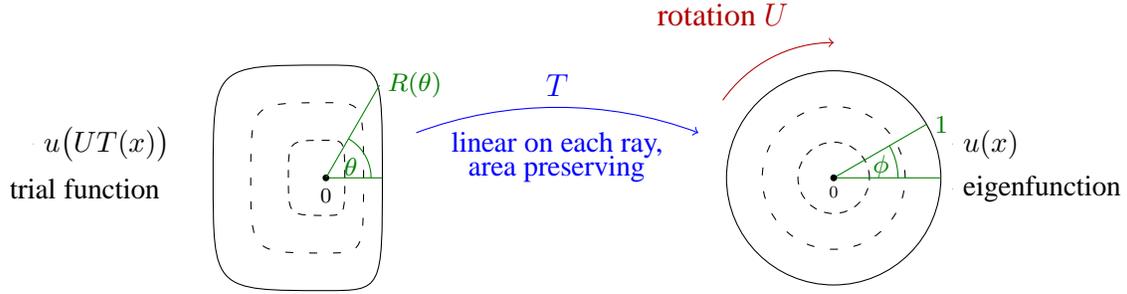
\begin{figure}
\begin{center}
\begin{tikzpicture}[scale=1.5]
  \begin{scope}[xshift=5cm]
        \draw[red!70!black,->] (145:1.2) node[above=24pt] {rotation $U$} arc (145:90:1.2);
    \draw (0,0) circle (0.95);
    \fill (1.05,0.3) circle (0) node [right] {\small $u(x)$};
        \fill (1.05,-0.1) circle (0) node [right] {\small eigenfunction};
    \draw[loosely dashed] (0,0) circle (0.95*2/3);
    \draw[dashed] (0,0) circle (0.95/3);
    \draw [green!50!black] (0,0) -- (0.95,0);
    \begin{scope}
      \draw [green!50!black](0,0) -- (30:0.95) node [right=-2pt,pos=1,scale=1.2] {\tiny $1$};
      \clip (0,0) -- (30:0.95) -- (0.95,0) --cycle;
      \draw [green!50!black] (0,0) circle (0.57);
      \draw[green!50!black] (0.42,0.1) node[scale=1.2] {\tiny $\phi$};
    \end{scope}
    \fill (0,0) circle (0.03) node [below,scale=0.8] {\tiny $0$};
  \end{scope}
    \draw[->,blue] (1.3,0.4) .. controls (2.1,0.7) and (3,0.7) .. (3.8,0.4) node [above,pos=0.5] {$T$} node[below=5pt,pos=0.5] {\small linear on each ray,} node[below=14pt,pos=0.5] {\small area preserving};
  \begin{scope}[xscale=-1,xshift=-5.5cm]
    \draw[smooth] (4.5,0) .. controls (4.5,1) .. (5,1) .. controls (6,1) .. (6,0) .. controls (6,-1) .. (5,-1) .. controls (4.5,-1) .. (4.5,0);
    \draw[smooth,loosely dashed,xshift=-5cm,scale=2/3,xshift=10cm] (4.5,0) .. controls (4.5,1) .. (5,1) .. controls (6,1) .. (6,0) .. controls (6,-1) .. (5,-1) .. controls (4.5,-1) .. (4.5,0);
    \draw[smooth,dashed,xshift=-5cm,scale=1/3,xshift=25cm] (4.5,0) .. controls (4.5,1) .. (5,1) .. controls (6,1) .. (6,0) .. controls (6,-1) .. (5,-1) .. controls (4.5,-1) .. (4.5,0);
  \begin{scope}[xscale=-1,xshift=-10cm]
    \draw[green!50!black] (5,0) -- (5.5,0);
    \fill (5,0) circle (0.03) node [below] {\tiny $0$};
        \fill (2.4,0.3) circle (0) node [right] {\small $u \big( UT(x) \big)$};
         \fill (2.1,-0.1) circle (0) node [right] {\small trial function};
   \draw[green!50!black] (5,0) -- ++(60:0.95) node [right=-2pt,pos=1,scale=1.2] {\tiny $R(\theta)$};
    \clip (5,0) -- +(1,0) -- +(60:1) -- cycle;
    \draw[green!50!black] (5,0) circle (0.4);
    \draw[green!50!black] (5.22,0.1) node [scale=1.2] {\tiny $\theta$};
  \end{scope}
  \end{scope}
\end{tikzpicture}
\end{center}
   \vspace{-0.3cm}
\caption{A linear-on-rays transformation from a domain $\Omega$ of area $\pi$ to the unit disk. To insure that the mapping preserves area locally, we require $R(\theta)^2 \, d\theta = d\phi$.}
\label{fig:transplant}
\end{figure}

\subsection*{Perturbations of the disk}
Let us apply the theorem to the ground state energy of a nearly circular domain. Suppose $P(\theta)$ is a Lipschitz continuous, $2\pi$-periodic function with Fourier series  
\[
P(\theta) = \sum_{n \in \Z} p_n e^{in\theta} ,
\]
where $p_{-n}=\overline{p_n}$ since $P$ is real-valued. Define a plane domain $\Omega_\e = \{ re^{i\theta} : 0 \leq r < 1+\e P(\theta) \}$, and assume $\e$ is small enough that the radius $1+\e P(\theta)$ is positive for all $\theta$. Obviously $\Omega_\e$ is a perturbation of the unit disk $\D$, when $\e$ is small. Write 
\[
\lambda_\e=\lambda_1(\Omega_\e)
\]
for the first Dirichlet eigenvalue of the magnetic Laplacian on $\Omega_\e$. Let $A_\e$ be the area of the domain, and remember that the flux through each domain $\Omega_\e$ is the same, namely $\beta$.
\begin{corollary}[Nearly circular domains] \label{perturb2}
The first magnetic eigenvalue of the domain $\Omega_\e$ is bounded above and below in terms of the boundary perturbation, with 
\[
1 \leq \frac{\lambda_\e A_\e}{\lambda_1(\D) \pi}
\leq 1 + 2 \e^2 \max \left\{ \sum_{n=1}^\infty n^2 |p_n|^2 \, , \, 4 \sum_{n=1}^\infty |p_n|^2 \right\} + O(\e^3) = G(\Omega_\e)
\]
as $\e \to 0$ with $P$ fixed.
\end{corollary}
The lower bound in the corollary is the Faber--Krahn type inequality due to Erd\"{o}s \cite{E96}. The upper bound follows from \autoref{Deigen}, as we show in \autoref{sec:formal}. 

To continue our investigation of nearly circular domains, we perform a formal perturbation analysis on the ground state energy. Write $M(a,b,z)$ for the Kummer function, also known as the confluent hypergeometric function  \cite[Chapter 13]{NIST}, and denote its $z$-derivative by $M^\prime$. For $\beta > 0$ we let
\[
a_0 = a_0(\beta) = \frac{1}{2} (1-\lambda_0 \pi/\beta) .
\]
\begin{theorem}[Perturbation] \label{th:perturb}
Fix $\beta > 0$ and assume $p_0=0$. Perturbation analysis yields the following formal asymptotic series for $\lambda_\e A_\e$ as $\e \to 0$ (with $P$ fixed):
\begin{align}
\lambda_\e A_\e = \lambda_1(\D) \pi +\Big( \sum_{n=2}^\infty c |p_n|^2 q_n \Big)  \e^2 + O(\e^3) , \label{eq:pertseries}
\end{align}
where
\begin{align*}
z & = \frac{\beta}{2\pi} , \\
c & = \frac{- M^\prime(a_0,1,z)}{\frac{\partial M}{\partial a}(a_0,1,z)} \, \frac{4\beta^2}{\pi} , \\
q_n & = 1+n -z + z ( \log M )^\prime(a_0,n+1,z) + z ( \log M )^\prime(a_0+n,n+1,z) .
\end{align*}
\end{theorem}
The assumption $p_0=0$ is harmless, since it essentially amounts to a rescaling of the perturbed domain. The positivity assumption on $\beta$ is for convenience only, and imposes no genuine restriction since $-\beta$ and $+\beta$ yield the same eigenvalues (the energy levels are independent of the direction of the magnetic field). 

We prove \autoref{th:perturb} in \autoref{sec:formal}, and show there that $q_n = n+O(1)$ as $n \to \infty$. Thus the $\e^2$-term in the asymptotic series \eqref{eq:pertseries} involves the $H^{1/2}$-norm of the boundary perturbation, while the second order term in \autoref{perturb2} is essentially the $H^1$-norm. Hence the asymptotic formula is in that sense sharper, although on the other hand we have no control over its error term. It is an open problem to prove an inequality (or error estimate) that captures the asymptotic series to second order. This problem is open even for the Laplacian ($\beta=0$). 

\begin{remark}\rm 
The summation in \eqref{eq:pertseries} begins with $n=2$, which leads one to ask: might the first eigenvalue actually decrease under boundary perturbations of type $n=1$, that is, $\cos \theta$ or $\sin \theta$? No! The ground state energy increases under such perturbations, as follows from the magnetic Faber--Krahn result of Erd\"{o}s (the lower bound in \autoref{perturb2}). This observation highlights the subtlety of Erd\"{o}s's result, and of the original Faber--Krahn theorem in the nonmagnetic case ($\beta=0$). 
\end{remark}

\subsection*{Neumann boundary conditions}
Assume the magnetic field is nonzero in what follows, meaning $\beta \neq 0$. (The zero field case was treated in the earlier paper \cite{LS13}.) Write $\{ \mu_j \}$ for the Neumann eigenvalues of the magnetic Laplacian on $\Omega$, so that the corresponding $L^2$-orthonormal eigenfunctions $u_j$ satisfy
\[
\begin{cases}
(i \nabla + F)^2 u_j = \mu_j u_j \;\; \text{in $\Omega$} \\
\hfill \vec{n}  \cdot (\nabla-iF)u = 0 \;\; \text{on $\partial \Omega$}
\end{cases}
\]
where as before, $F(x)=\frac{\beta}{2A}(-x_2,x_1)$. The boundary condition arises naturally from minimization of the Rayleigh quotient, and it plays no role in our proofs. The eigenvalues satisfy
\[
0 < \mu_1 \leq \mu_2 \leq \mu_3 \leq \dots 
\]
where we note that positivity of the first eigenvalue holds because the field is nonzero (see, for example, \cite[Lemma A.8]{LLR12}). 
\begin{theorem}[Neumann magnetic Laplacian] \label{Neigen}
Assume $\beta \neq 0$. Then \autoref{Deigen} holds with Dirichlet eigenvalues replaced by Neumann eigenvalues, except omitting the equality statement from the theorem. 
\end{theorem}
The proof goes exactly as for the Dirichlet case in \autoref{Deigen}, except using the trial function space $H^1(\Omega;\C)$ rather than $H^1_0(\Omega;\C)$. The proof of the equality statement breaks down, because the Neumann ground state need not be radial. 

For the disk, the magnetic eigenvalues with Neumann boundary conditions can in principle be computed in terms of Kummer functions, although in practice the equations become rather complicated. The eigenvalue branches display fascinating behavior. For example, a numerical study due to Saint-James \cite{SJ65} reveals that the Neumann ground state has angular dependence $e^{in(\beta)\theta}$, where the number $n(\beta)$ increases to infinity as the flux $\beta$ increases to infinity.

\subsection*{Relevant literature, and the contributions of this paper} Few isoperimetric type inequalities are known for magnetic eigenvalues. This paucity stands in stark contrast to the rich body of work developed for  the nonmagnetic Laplacian over the past century, for which one may consult the surveys by Ashbaugh and Benguria \cite{AB07} or Benguria and Linde \cite{BL08}, and the monographs of Bandle \cite{B80}, Henrot \cite{He06}, Kesavan \cite{K06} and P\'{o}lya--Szeg\H{o} \cite{PS51}. The main contribution of this paper is to prove the first known sharp upper bounds for magnetic spectral functionals on more-or-less general plane domains. 

This paper generalizes our earlier work on eigenvalues of the Laplacian \cite{LS13}, that is, on the case of zero magnetic field. Those earlier results hold in all dimensions, with balls as the maximizers. We restrict in this paper to plane domains, because in three dimensions and higher, introducing a magnetic field creates a preferred direction in the problem, breaking the symmetry and rendering our proof invalid. 

An advantage of working only in the plane is that we can develop a significantly simpler approach than in higher dimensions. The proof of our main result, \autoref{Deigen}, relies on the special fact that rotations commute in two dimensions: we exploit this fact to construct a proof that is both shorter and easier to understand than in our earlier work on the Laplacian. Thus for readers who are new to this subject, we recommend beginning with the proof of \autoref{Deigen} in the zero-field case, taking $\beta=0$ throughout the proof, and only then turning to the magnetic case ($\beta \neq 0$) or to the higher dimensional case in our earlier paper \cite{LS13}. 

Another work to which this current paper owes a debt is that of Laugesen, Liang and Roy \cite{LLR12}. They treated a restricted class of domains, namely linear images of rotationally symmetric domains such as regular polygons, and obtained sharp upper bounds on magnetic eigenvalue sums with the maximizing domains being the original rotationally symmetric domains. For example, the centered equilateral triangle was shown to maximize the eigenvalue sum $(\lambda_1 + \cdots + \lambda_n)A/G_1$ among all triangles. (The authors could have subsequently invoked majorization to pass to spectral functionals such as the partition function, like in this paper, but did not do so.)

One difference between the work of Laugesen \emph{et al.}\ and the current paper is that here we average over the full group of rotations instead of over discrete subgroups such as the $3$-fold rotations for the equilateral triangle. Thus we avoid the tight frame theory that was needed in the earlier paper \cite{LLR12}. Another difference is that the transformations in that paper were rather simple (in fact, globally linear), whereas in the current paper we must use more complicated linear-on-rays transformations such as shown in \autoref{fig:transplant}, in order to map the disk to general starlike domains. This additional complexity forces the inclusion of the boundary oscillation factor $G_0$ in our theorems; it was not needed in the earlier work, since a linear transformation stretches without oscillation. 

Lastly we mention some inequalities related to semi-classical constants. Our results in this paper can be called ``geometrically sharp'', since an extremal domain exists for each spectral functional. 
We call a spectral inequality ``asymptotically sharp''  if it holds with equality in the limit $n \to \infty$, for each domain. An asymptotically sharp inequality of Berezin--Li--Yau type holds for magnetic eigenvalue sums, by work of Erd\"{o}s, Loss and Vougalter \cite{ELV00}, extending results of Laptev and Weidl \cite{LW00}. See also a later work of Frank, Laptev and Molchanov \cite{FLM08}. In the negative direction, the magnetic P\'{o}lya conjecture was disproved by Frank, Loss and Weidl \cite{FLW09} by constructing a counterexample from square domains. 

\subsection*{Open problems}
Erd\"{o}s proved under Dirichlet boundary conditions that the magnetic ground state energy is minimal for a disk of the same area \cite{E96}. In scale invariant terms, he proved $\lambda_1 A$ is minimal for the disk. This result suggests several open problems. 

Is the scale invariant magnetic partition function $\sum_{j=1}^\infty e^{-\lambda_j A t}$ maximal for the disk, for each $t>0$? Luttinger \cite{L73} proved the result for the Laplacian ($\beta=0$). Note that letting $t \to \infty$ would recover the minimality of the first eigenvalue. 
 
Next, for the Neumann spectrum does one have minimality of $\mu_1 A$ for the disk? This conjecture holds trivially for the Laplace operator, because there $\mu_1=0$ for all domains. Thus one should begin by investigating the conjecture for small values of $\beta \neq 0$ (small nonzero fields) using a perturbation analysis. Even if the conjecture holds for such $\beta$ values, it might fail when $\beta$ is larger because the nature of the ground state changes as $\beta$ increases: the Neumann ground state of the disk is radial for small values of $\beta$ but has angular dependence when $\beta$ is large, as was found numerically by Saint-James \cite{SJ65}.

\section{\bf Results for the Pauli operator}
\label{sec:pauli}

To study a charged particle with spin $1/2$, we investigate the energy levels of the Pauli operator $H_P = \big( \sigma \cdot (i\nabla + F)  \big)^2$ (see \cite{E07}, for example). Here $\sigma=(\sigma_1,\sigma_2,\sigma_3)$ is the $3$-tuple of self-adjoint Pauli matrices:
\[
\sigma_1 = \begin{pmatrix} 0 & 1 \\ 1 & 0 \end{pmatrix} , \qquad \sigma_2 = \begin{pmatrix} 0 & -i \\ i & 0 \end{pmatrix} , \qquad \sigma_3 = \begin{pmatrix} 1 & 0 \\ 0 & -1 \end{pmatrix} . 
\]
The Pauli operator acts on spinors, that is, on $2$-component complex vector fields of the form $\psi=\begin{pmatrix} \psiplus \\ \psiminus \end{pmatrix}$. For planar motion with a perpendicular magnetic field $(0,0,\beta/A)$, we may assume the wavefunction $\psi$ is independent of $x_3$ and that the gradient $\nabla=(\partial_1,\partial_2)$ and vector potential $F=(F_1,F_2)=(-x_2,x_1) \beta/2A$ have only  two components. Thus the planar Pauli operator is the formally self-adjoint operator
\[
H_P = \Big( \sum_{m=1}^2 \sigma_m (i\partial_m + F_m) \Big)^2 
\]
with corresponding Rayleigh quotient
\begin{equation} \label{RayPauli}
\Ray_P[\psi] = \frac{\int_\Omega \big| \sum_{m=1}^2 \sigma_m (i\partial_m + F_m) \psi \big|^2 \, dx}{\int_\Omega |\psi|^2 \, dx} . 
\end{equation}
The numerator is known to be elliptic for $\psi \in H^1_0(\Omega;\C^2)$, meaning it is bounded below by a constant times $\int_\Omega |\nabla \psi|^2 \, dx$ minus a constant times $\int_\Omega |\psi|^2 \, dx$. (A brief demonstration of ellipticity is included in \autoref{Peigen_proof}.) Hence the Dirichlet spectrum of the Pauli operator is discrete, by the spectral theorem for quadratic forms. We place the eigenvalues in increasing order, so that
\[
0 < \lambda_1^P \leq \lambda_2^P \leq \lambda_3^P \leq \dots 
\]
where ``$P$'' stands for Pauli. Positivity of the first eigenvalue will be justified in \autoref{Peigen_proof}.

The next theorem provides sharp upper bounds on shifted Dirichlet eigenvalues of the Pauli operator. 
\begin{theorem}[Dirichlet-Pauli operator] \label{Peigen}
\autoref{Deigen} holds with the Dirichlet eigenvalues of the magnetic Laplacian replaced by suitably shifted Dirichlet eigenvalues of the Pauli operator; specifically, one replaces $\lambda_j$ with $\lambda^P_j+|\beta|/A$.
\end{theorem}
We do not have an analogous result for the Neumann eigenvalues of the Pauli operator. Indeed, the Pauli operator does not have discrete spectrum on $H^1$, because its null space is infinite dimensional; see the discussion in \autoref{Peigen_proof}.

\section{\bf The constant Jacobian transformation}
\label{sec:setup}

In \autoref{fig:transplant} we showed how to construct a mapping from a starlike domain of area $\pi$ to the unit disk: we choose the map to be linear on each ray, with the angular deformation of rays determined by requiring that the mapping should preserve the area of each infinitesimal sector.

For a general starlike domain, we simply rescale the formula from \autoref{fig:transplant}. More precisely, we determine the angular deformation $\phi(\theta)$ by integrating the initial value problem
\[
\phi^\prime(\theta) = R(\theta)^2 \frac{\pi}{A} , \qquad \phi(0)=0 .
\]
Notice $\phi$ increases by $2\pi$ as $\theta$ increases by $2\pi$, since $\int_0^{2\pi} R(\theta)^2 \, d\theta = 2A$. Then we define a transformation 
\[
T : \Omega \to \D
\]
in polar coordinates by 
\[
T(r,\theta) = \big( r/R(\theta),\phi(\theta) \big) .
\]
Obviously the transformation is linear with respect to $r$, on each ray, and one easily checks that the Jacobian is constant, with $\text{Jac}(T) \equiv \pi/A$.

\section{\bf Dirichlet eigenvalues --- proof of \autoref{Deigen}}
\label{Deigen_proof}

The Rayleigh principle characterizes eigenvalues in terms of minima over classes of trial functions, and so in order to get upper bounds, our task is to choose suitable trial functions. We will construct trial functions on $\Omega$ by transplanting eigenfunctions from the disk with the help of the area-preserving map $T$ constructed in the previous section. (Note this method does not require explicit formulas for the eigenfunctions on the disk.) Then we average with respect to all pre-rotations of the disk.  

The Rayleigh quotient associated with the Dirichlet spectrum \eqref{eq:specD} of the magnetic Laplacian is
\[
\Ray[v] = \frac{\int_\Omega \big| (i\nabla + \frac{\beta}{2A}(-x_2,x_1))v \big|^2 \, dx}{\int_\Omega |v|^2 \, dx} \qquad \text{for\ } v \in H^1_0(\Omega;\C) . 
\]
Expressing the numerator in polar coordinates (writing $\mathbf{e}_r$ and $\mathbf{e}_\theta$ for the radial and angular unit vectors) gives that 
\begin{align}
\Ray[v] 
& = \frac{\int_\Omega \big| i v_r \mathbf{e}_r + i r^{-1} v_\theta \mathbf{e}_{\theta} + \frac{\beta}{2A} v r \mathbf{e}_\theta \big|^2 \, dx}{\int_\Omega |v|^2 \, dx} \notag \\
& = \frac{\int_\Omega \big\{ |v_r|^2 + |i r^{-1} v_\theta + \frac{\beta}{2A}rv|^2 \big\} \, rdrd\theta}{\int_\Omega |v|^2 \, rdrd\theta} . \label{eq:Rayleighpolar}
\end{align}

The Rayleigh--Poincar\'{e} Variational Principle \cite[p.~98]{B80} characterizes the sum of the first $n$ Dirichlet eigenvalues as:
\begin{align*}
\lambda_1 + \dots + \lambda_n & = \min \big\{ \Ray[v_1] + \dots + \Ray[v_n] : \\
& \qquad v_1, \dots,v_n \in H^1_0(\Omega;\C)\text{\ are pairwise orthogonal in $L^2(\Omega;\C)$} \big\} .
\end{align*}
To apply this principle, we let $u_1,u_2,u_3,\ldots$ be orthonormal eigenfunctions on the unit disk $\D$ corresponding to the  eigenvalues $\lambda_1(\D),\lambda_2(\D),\lambda_3(\D),\ldots$. Let $\eta \in \R$ and use $U$ to denote rotation of the plane by angle $\eta$. Then define trial functions on $\Omega$ by 
\[
v_j = u_j \circ U^{-1} \circ T
\]
where the transformation $T : \Omega \to \D$ was defined in \autoref{sec:setup}. Thus in polar coordinates we have
\begin{equation} \label{eq:uv}
v_j(r,\theta) = u_j \big( r/R(\theta) , \phi(\theta) - \eta \big) .
\end{equation}

One can show that $v_j \in H^1 \cap C(\Omega)$, by using that $R(\theta)$ is Lipschitz and $\phi(\theta)$ is continuously differentiable. Further, $v_j=0$ on the boundary of $\Omega$ because $u_j=0$ on the boundary of the disk. Thus $v_j \in H^1_0(\Omega)$.

The functions $v_j$ are pairwise orthogonal, since
\begin{align}
\int_\Omega v_j \overline{v_k} \, dx 
& = \text{Jac} (T^{-1}) \int_\D u_j \overline{u_k} \, dx \label{eq:pairorthog} \\
& = 0 \notag
\end{align}
whenever $j \neq k$, using here that $u_j$ and $u_k$ are orthogonal and $T^{-1}$ has constant Jacobian. Thus by the Rayleigh--Poincar\'{e} principle, we have
\begin{equation} \label{eq:rayest}
\sum_{j=1}^n \lambda_j(\Omega) \leq \sum_{j=1}^n \frac{\int_\Omega \big| (i\nabla + \frac{\beta}{2A}(-x_2,x_1))v_j \big|^2 \, dx}{\int_\Omega |v_j|^2 \, dx} .
\end{equation}

The denominator of this Rayleigh quotient is $\int_\Omega |v_j|^2 \, dx = \text{Jac}(T^{-1}) = A/\pi$ by \eqref{eq:pairorthog} with $j=k$, since the eigenfunctions are normalized with $\int_\D |u_j|^2 \, dx=1$.

To evaluate the numerator of the Rayleigh quotient, we develop some lemmas. Write $u=u_j$ and $v=v_j$, to simplify notation in what follows, and express $u$ and $v$ in polar coordinates as $u(s,\phi)$ and $v(r,\theta)$, respectively. These functions are related by \eqref{eq:uv}. 
\begin{lemma} \label{le:radial}
\begin{align*}
|v_r(r,\theta)|^2 R(\theta)^2 = | u_s \big( r/R(\theta), \phi(\theta) - \eta \big)|^2 .
\end{align*}
\end{lemma}
\begin{proof}
Simply differentiate \eqref{eq:uv} with respect to $r$, and square the result. 
\end{proof}
\begin{lemma} \label{le:angular}
\begin{align*}
& \left| i r^{-1} v_\theta + \frac{\beta}{2A}rv \right|^2 R(\theta)^2 \\
& = |u_s(r/R(\theta),\phi(\theta)-\eta)|^2 \, (\log R)^\prime(\theta)^2 
\ + \ 2 \Real \, \overline{u_s(r/R(\theta),\phi(\theta)-\eta)} \quad \times \\
& \qquad \Big( - \frac{R(\theta)}{r} u_\phi(r/R(\theta),\phi(\theta)-\eta) + \frac{i \beta}{2\pi} \frac{r}{R(\theta)} u(r/R(\theta),\phi(\theta)-\eta) \Big) \frac{\pi}{A} R(\theta) R^\prime(\theta) \\
& + \left| i \frac{R(\theta)}{r} u_\phi(r/R(\theta),\phi(\theta)-\eta) + \frac{\beta}{2\pi} \frac{r}{R(\theta)} u(r/R(\theta),\phi(\theta)-\eta) \right|^2 \, \frac{\pi^2 R(\theta)^4}{A^2} .
\end{align*}
\end{lemma}
\begin{proof}
Differentiating \eqref{eq:uv} with respect to $\theta$ gives that
\[
v_\theta(r,\theta) = - r u_s \big( r/R(\theta), \phi(\theta) - \eta \big) R^\prime(\theta)/R(\theta)^2 + u_\phi \big( r/R(\theta) , \phi(\theta) - \eta \big) \phi^\prime(\theta) .
\]
Substituting this formula into the left side of the lemma yields an expression of the form $|a+b+c|^2$, which we expand as $|a|^2 + 2 \Real \overline{a}(b+c) + |b+c|^2$, hence obtaining the right side of the lemma. In the final simplification we use also that $\phi^\prime(\theta) = R(\theta)^2 \pi/A$. 
\end{proof}
\begin{lemma} \label{le:polar}
The numerator of the Rayleigh quotient for $v$  is
\[
\int_\Omega \left| \left(i\nabla + \frac{\beta}{2A}(-x_2,x_1)\right)v \right|^2 \, dx = Q_1 + Q_2 + Q_3
\]
where 
\begin{align*}
Q_1 & = \int_0^{2\pi} \int_0^1 \big| u_s( s,\phi(\theta) - \eta) \big|^2 \, s ds \, \big[ 1 + (\log R)^\prime(\theta)^2 \big]  \, d\theta , \\
Q_2 & = 2 \Real \int_0^{2\pi} \int_0^1 \overline{u_s(s,\phi(\theta)-\eta)} \quad \times \\
& \hspace*{3cm}  \Big( - s^{-1} u_\phi(s,\phi(\theta)-\eta) + \frac{i \beta}{2\pi} s u(s,\phi(\theta)-\eta) \Big) \, s ds \, \frac{\pi}{A} R(\theta) R^\prime(\theta) \, d\theta , \\
Q_3 & = \int_0^{2\pi} \int_0^1 \left| i s^{-1} u_\phi(s,\phi(\theta)-\eta) + \frac{\beta}{2\pi} s u(s,\phi(\theta)-\eta) \right|^2 \, s ds \frac{\pi^2 R(\theta)^4}{A^2} \, d\theta .
\end{align*}
\end{lemma}
\begin{proof}
Start with the numerator in polar coordinates as in \eqref{eq:Rayleighpolar}, then substitute using \autoref{le:radial} and \autoref{le:angular}, and make the radial change of variable $r=sR(\theta)$, so that $0<s<1$.  
\end{proof}
\begin{lemma} \label{le:polaraverage}
The averages of $Q_1,Q_2,Q_3$ with respect to $\eta$ are:
\begin{align*}
\frac{1}{2\pi} \int_0^{2\pi} Q_1 \, d\eta & = G_0(\Omega) \int_\D |u_s|^2 \, dx , \\
\frac{1}{2\pi} \int_0^{2\pi} Q_2 \, d\eta & = 0 , \\
\frac{1}{2\pi} \int_0^{2\pi} Q_3 \, d\eta & = G_1(\Omega) \int_\D \left| i s^{-1} u_\phi + \frac{\beta}{2\pi} su \right|^2 \, dx .
\end{align*}
\end{lemma}
\begin{proof}
For $Q_1$, we integrate the definition in \autoref{le:polar} with respect to $\eta$ and interchange the order of integration. Then making the substitution $\eta \mapsto \phi(\theta)-\eta$ allows us to separate the $\eta$ and $\theta$ integrals, which completes the proof when we recall the definition of $G_0$ from the Introduction. The argument is analogous for $Q_3$.

With $Q_2$ we proceed similarly, and then observe that $\int_0^{2\pi} R(\theta) R^\prime(\theta) \, d\theta = \frac{1}{2} R(\theta)^2 \big|^{2\pi}_0 = 0$ by periodicity. 
\end{proof}
\noindent [\emph{Aside, not needed in the rest of the paper:} Our proof that the cross-term $Q_2$ vanishes after averaging with respect to $\eta$ seems like a trick since it relies on the quantity $RR^\prime$ being a derivative. To avoid using this fact, one may include a reflection as well as rotations when constructing trial functions, as follows. Write $\Pi$ for reflection in the horizontal axis, and consider the additional trial functions 
\[
w_j = \overline{u_j} \circ \Pi \circ U^{-1} \circ T ,
\]
which in polar coordinates can be written
\[
w_j(r,\theta)=\overline{u_j(r/R(\theta),\eta-\phi(\theta))}  .
\]
Now carry out the proof as above, except using $w_j$ instead of $v_j$. The resulting quantities $Q_1$ and $Q_3$ are the same as for $v_j$, but $Q_2$ acquires a negative sign in front. Hence by averaging the numerators of the Rayleigh quotients for $v_j$ and $w_j$ we eliminate $Q_2$ and obtain simply $Q_1+Q_3$. Then one  averages with respect to $\eta$ by the formulas in \autoref{le:polaraverage}.]

\medskip
Now we return to the proof of the theorem. The left side of \eqref{eq:rayest} is independent of the rotation angle $\eta$. Hence by averaging \eqref{eq:rayest} with respect to $\eta \in [0,2\pi]$ we find
\[
\sum_{j=1}^n \lambda_j(\Omega) \leq \sum_{j=1}^n \frac{\frac{1}{2\pi} \int_0^{2\pi} (Q_1+Q_2+Q_3) \, d\eta}{A/\pi} ,
\]
where we must remember that ``$u$'' means $u_j$, in the quantities $Q_1,Q_2,Q_3$. Thus \autoref{le:polaraverage} shows that 
\begin{align*}
\sum_{j=1}^n \lambda_j(\Omega) A & \leq \pi \sum_{j=1}^n \left[ G_0(\Omega) \int_\D \left| \frac{\partial u_j}{\partial s} \right|^2 \, dx + G_1(\Omega) \int_\D \left| i s^{-1} \frac{\partial u_j}{\partial \phi} + \frac{\beta}{2\pi} su_j \right|^2 \, dx \right] \\
& = \pi \sum_{j=1}^n \Big[ (1-\alpha_j) G_0(\Omega) + \alpha_j G_1(\Omega) \Big] \int_\D \left| (i\nabla + \frac{\beta}{2\pi}(-x_2,x_1))u_j \right|^2 \, dx
\end{align*}
where 
\[
\alpha_j = \frac{\int_\D \big| i s^{-1} \frac{\partial u_j}{\partial \phi} + \frac{\beta}{2\pi} su_j \big|^2 \, dx}{\int_\D \big| (i\nabla + \frac{\beta}{2\pi}(-x_2,x_1))u_j \big|^2 \, dx} , \qquad j=1,2,\ldots,n .
\]
The coefficient $\alpha_j \in [0,1]$ measures the ``angular component'' of the magnetic energy of the $j$th mode; see \eqref{eq:Rayleighpolar}.

We may estimate $G_0$ and $G_1$ from above with their maximum, $G$, so that
\[
\sum_{j=1}^n \lambda_j(\Omega) A(\Omega)/G(\Omega) \leq \pi \sum_{j=1}^n \int_\D \left| \left(i\nabla + \frac{\beta}{2\pi}(-x_2,x_1)\right)u_j \right|^2 \, dx = \pi \sum_{j=1}^n \lambda_j(\D) .
\]
Since $A(\D)=\pi$ and $G(\D)=1$, the theorem is proved for the case that $\Phi(a) \equiv a$ is the identity function. Note we have proved the result with a unit disk on the right side, but any centered disk will do, by scale invariance. 

Now Hardy--Littlewood--P\'{o}lya majorization extends the result to all concave increasing $\Phi$. (For references on majorization, see \cite[Appendix~A]{LS13}.)

\smallskip
\noindent \emph{Remark.} Before completing the proof, we ask: could the factor $G_1$ be eliminated from the theorem when we study the Dirichlet ground state energy $\lambda_1$? The answer is No in the magnetic case, because even though the ground state of the disk is purely radial (a fact which seems to be non-obvious \cite{S14}), one has $\alpha_1>0$ by the definition above, when $\beta$ is nonzero. 

\subsection*{Particular choices of $\Phi$.} The function $\Phi(a)=a^s$ is concave and increasing, when $0 < s \leq 1$, and this choice of $\Phi$ gives maximality of $(\lambda_1^s + \cdots + \lambda_n^s)^{1/s} \, A/G$ for the centered disk. Choosing $\Phi(a)=\log a$ shows maximality of the centered disk for the functional
\[
\sum_{j=1}^n \log (\lambda_j A/G) = n \log \Big( \sqrt[n]{\lambda_1 \cdots \lambda_n} \, A/G \Big) .
\]
The function $\Phi(a)= - a^s$ is concave increasing, when $s<0$, and so we obtain minimality of the centered disk for $\sum_{j=1}^n (\lambda_j A/G)^s$. Lastly, for $t>0$ we consider $\Phi(a)=-e^{-at}$ to prove minimality of $\sum_{j=1}^n \exp(-\lambda_j At/G)$ for the centered disk.

\subsection*{Dirichlet equality statement.} Assume equality holds for the first eigenvalue, that is,
\[
\lambda_1 A/G \big|_\Omega = \lambda_1 A /G \big|_\D .
\]
By enforcing equality in our proof above, with $n=1$, we see that the trial function $v_1$ on $\Omega$ must attain equality in the Rayleigh characterization of $\lambda_1(\Omega)$, and hence must be a first eigenfunction for $\Omega$. In particular this holds when $\eta=0$ (no rotation), so that the function $v(x)=u_1(T(x))$ satisfies
\[
(i\nabla+F)^2 v = \lambda_1(\Omega) v 
\]
classically, where $F(x)=\frac{\beta}{2A}(-x_2,x_1)$. That is, 
\begin{equation} \label{eq:expanded}
-\Delta v+2iF\cdot \nabla v+|F|^2 v = \lambda_1(\Omega) v . 
\end{equation}

The ground state of the disk is radial for the magnetic Laplacian with Dirichlet boundary conditions \cite{S14}, so that $u_1(x)=J(|x|)$ for some real-valued function $J$. (In the case of the Laplacian this $J$ is simply the zeroth Bessel function, whereas for the magnetic Laplacian it is a Kummer function, as discussed in the next section.) Observe that $J^\prime(r_0) \neq 0$ for some $r_0 \in(0,1)$ because $J$ cannot be identically zero and $J(1)=0$ by the Dirichlet boundary condition. 

Taking imaginary parts in \eqref{eq:expanded} shows that $F\cdot \nabla v = 0$, so that $\beta v_\theta = 0$. Suppose $\beta \neq 0$. Since $v(x) = u_1(T(x)) = J(r/R(\theta))$, we have
\begin{align*}
  0=v_\theta=-J^\prime\big(r/R(\theta)\big)rR(\theta)^{-2} R^\prime(\theta) .
\end{align*}
Choosing $r=r_0 R(\theta)$, we deduce that $R^\prime(\theta)=0$ for almost every $\theta$. Hence the radius function is constant, which means $\Omega$ is a centered disk.

For the equality case when $\beta = 0$, see our earlier work \cite[Theorem 3.1]{LS13}. That earlier work assumes $R(\theta)$ is $C^2$-smooth, but in fact that smoothness follows from inverting the formula $v(x) = J(r/R(\theta))$ to solve for $R$, using smoothness of the first eigenfunction $v$ and the radial function $J$.

\section{\bf Perturbation analysis}
\label{sec:formal}

\subsection*{Proof of \autoref{perturb2}}
Start by applying \autoref{Deigen} to the first eigenvalue ($n=1$) and then substitute $R=1+\e P$ into the definitions of $G_0$ and $G_1$. One obtains the following expressions:
\begin{align*}
G_0(\Omega_\e) = 1+\int_0^{2\pi} \frac{\e^2 P^\prime(\theta)^2}{\big( 1+\e P(\theta) \big)^2} \, \frac{d\theta}{2\pi}  
& = 1 + \e^2 \int_0^{2\pi} P^\prime(\theta)^2 \, \frac{d\theta}{2\pi} + O(\e^3) \\
& = 1 + \e^2 \sum_{n \neq 0} n^2 |p_n|^2 + O(\e^3)
\end{align*}
and
\begin{align*}
G_1(\Omega_\e) = \frac{\int_0^{2\pi} \big( 1 + \e P(\theta)\big)^{\! 4} \, d\theta/2\pi}{\big[ \int_0^{2\pi} \big( 1 + \e P(\theta)\big)^{\! 2} \, d\theta/2\pi \big]^{\! 2}} 
& = 1 + 4 \e^2 \int_0^{2\pi} \big( P(\theta) - p_0 \big)^2 \, \frac{d\theta}{2\pi} + O(\e^3) \\
& = 1 + 4 \e^2 \sum_{n \neq 0} |p_n|^2 + O(\e^3) .
\end{align*}
The upper bound in the corollary now follows once we use the symmetry of the coefficients ($p_{-n}=\overline{p_n}$).

\subsection*{Proof of \autoref{th:perturb}}
The Kummer function $M(a,b,z)$ satisfies the differential equation
\begin{align}
  zM^{\prime \prime}+(b-z)M^\prime-aM = 0 \label{eq:kummerDE}
\end{align}
with initial condition $M(a,b,0)=1$. (Recall that primes indicate derivatives with respect to $z$.) Define
\[
  f_n (r,\lambda) = (r^2/\pi)^{|n|/2} e^{-\beta r^2/4\pi} M \left( (1+|n|-n - \lambda/\beta)/2 , |n|+1 , \frac{\beta r^2}{2\pi}\right)
\]
for $n \in \Z$ and $r,\lambda > 0$. We rescale $f_n$ by area,
then modulate by $e^{in \theta}$, and form a series combination as follows:
\[
u(r,\theta) = f_0( r\sqrt{\pi/A},\lambda A) +\e\sum_{n\ne 0} c_n f_{n}( r\sqrt{\pi/A},\lambda A) e^{in\theta} ,
\]
with coefficients $c_n$ to be chosen below. One finds that $u$ satisfies formally the eigenvalue equation $(i\nabla+F)^2 u=\lambda u$ in the plane, by expressing the eigenvalue equation in polar coordinates as 
\[
- (u_{rr}+r^{-1}u_r+r^{-2}u_{\theta \theta}) + i \frac{\beta}{A} u_\theta + \frac{\beta^2r^2}{4 A^2} u=\lambda u
\]
and then using the Kummer differential equation. 

The goal of the perturbation analysis is to choose the parameter $\lambda$ and coefficients $c_n$ so that $u=0$ on the boundary of the perturbed domain (at least to second order in $\e$). Then we have a Dirichlet eigenfunction, and it should be close to the ground state of the perturbed domain provided $\e$ is small enough; we do not seek to make these claims rigorous, since we are carrying out a formal analysis only. 

Recall $\Omega_\e$ is a nearly circular domain defined in polar coordinates by $r \le R(\theta)= 1+\e P(\theta)$, where $P$ is expressed in a Fourier series as before. We assume the constant term vanishes: $p_0=0$. Let
\[
\gamma =\sum_{n\ne0}|p_n|^2=\frac{1}{2\pi}\int_0^{2\pi} P(\theta)^2 \, d\theta
\]
so that the domain has area
\[
  A=\frac{1}{2} \int_0^{2\pi} R(\theta)^2 \, d\theta = \pi(1+\e^2\gamma)
\]
and hence
\begin{align*}
  \sqrt{\frac{\pi}{A}}&=1-\e^2\gamma/2+O(\e^4),\\
  R(\theta) \sqrt{\frac{\pi}{A}} &=1+\e P(\theta)-\e^2\gamma/2+O(\e^3) .
\end{align*}

Denote by $\lambda_0$ the lowest magnetic eigenvalue of the unit disk, with Dirichlet boundary condition. The corresponding ground state on the disk is a radial function $f_0(r,\lambda_0 \pi)$, satisfying
\begin{align} \label{eigeneq}
- (u_{rr}+r^{-1}u_r) + \frac{\beta^2r^2}{4 \pi^2} u=\lambda_0 u
\end{align}
with the boundary condition 
\begin{align}
  f_0(1,\lambda_0\pi)=0 , \qquad \text{or} \qquad M(a_0 , 1 , z) = 0 \label{eq:gsmag}
\end{align}
where $z=\beta/2\pi$; these claims about the ground state are justified in \cite{S14}. 

To carry out a perturbation analysis, we assume that the lowest eigenvalue of the perturbed domain varies with $\e$ according to
\[
  \lambda A=\lambda_0\pi+\rho \e +\tau \e^2 + O(\e^3) ,
\]
for some coefficients $\rho$ and $\tau$ to be determined. For the Dirichlet boundary condition on $\Omega_\e$ we require
\[
0 =  u(R(\theta),\theta) , \qquad \theta \in [0,2\pi] ,
\]
so that we want
\begin{align*}
 0 &=
  f_0\big( 1+\e P(\theta)-\e^2\gamma/2+O(\e^3),\lambda_0\pi+\rho \e+\tau \e^2 + O(\e^3) \big)+\\
  &
  \qquad+\e\sum_{n\ne0} c_n f_{n}\big( 1+\e P(\theta)-\e^2 \gamma/2+O(\e^3),\lambda_0\pi+\rho \e+\tau \e^2 + O(\e^3) \big) e^{in\theta} .
\end{align*}
Denote the partial derivatives of $f_n$ using superscripts, so that $f_n^{1,0}=\partial f_n/\partial r$. Then 
the Taylor expansion of the above boundary condition says to second order in $\e$ that
\begin{align*}
0 & =  f_0(1,\lambda_0\pi)
  +\e\sum_{n\ne0} c_n f_{n}(1,\lambda_0\pi) e^{in\theta}
  \\
  &\qquad+\left(\e P(\theta)-\frac{\e^2}{2}\gamma\right)f_0^{1,0}(1,\lambda_0\pi)+(\e \rho+\tau \e^2)f_0^{0,1}(1,\lambda_0\pi)
  \\
  &\qquad+\e\sum_{n\ne0} c_n \left(\e P(\theta)f_{n}^{1,0}(1,\lambda_0\pi)+\e \rho f_{n}^{0,1}(1,\lambda_0\pi)\right) e^{in\theta}
  \\
  &\qquad +
  \frac{1}{2}\e^2 P^2(\theta)f_0^{2,0}(1,\lambda_0 \pi)+\e^2\rho P(\theta)f_0^{1,1}(1,\lambda_0\pi)+
\frac{1}{2}\e^2\rho^2f_0^{0,2}(1,\lambda_0 \pi) + O(\e^3) .
\end{align*}
The zeroth order term vanishes by \eqref{eq:gsmag}. We want the first and the second order terms to vanish also. 

Vanishing of the first order term requires
\[
  \sum_{n\ne0} c_n f_{n}(1,\lambda_0\pi) e^{in\theta}+\sum_{n\ne0} p_n e^{in\theta}
  f_0^{1,0}(1,\lambda_0\pi)+\rho f_0^{0,1}(1,\lambda_0\pi)=0.
\]
The constant term in this equation tells us that $\rho f_0^{0,1}(1,\lambda_0\pi)=0$, and hence
\[
  \rho=0
\]
(assuming for now that $f_0^{0,1}(1,\lambda_0 \pi) \ne 0$, which we justify below). For $n\ne 0$, we get
\begin{align}\label{An}
  c_n=-p_n \frac{f_0^{1,0}(1,\lambda_0\pi)}{f_{n}(1,\lambda_0\pi)} .
\end{align}

If we average the second order term over $\theta$ and put $\rho=0$, we get
\begin{align*}
 \tau f_0^{0,1}(1,\lambda_0\pi)+\sum_{n\ne0} c_n p_{-n} f_{n}^{1,0}(1,\lambda_0\pi)
  +\frac{\gamma}{2}f_0^{2,0}(1,\lambda_0 \pi) -\frac{\gamma}2f_0^{1,0}(1,\lambda_0\pi)=0 .
\end{align*}
Note that $f_0^{2,0}(1,\lambda_0\pi)=-f_0^{1,0}(1,\lambda_0\pi)$ by evaluating the eigenvalue equation \eqref{eigeneq} at $r=1$. Hence using \eqref{An}, we may solve for the coefficient $\tau$ as
\begin{align}
  \tau 
& = \frac{f_0^{1,0}(1,\lambda_0\pi)}{f_0^{0,1}(1,\lambda_0\pi)} \sum_{n\ne0} |p_n|^2\left(1+ \frac{f_{n}^{1,0}(1,\lambda_0\pi)}{f_{n}(1,\lambda_0\pi)}\right) \notag \\
& = 2 \frac{f_0^{1,0}(1,\lambda_0\pi)}{f_0^{0,1}(1,\lambda_0\pi)} \sum_{n=1}^\infty |p_n|^2\left(1+ \frac{1}{2} \frac{f_{n}^{1,0}(1,\lambda_0\pi)}{f_{n}(1,\lambda_0\pi)} + \frac{1}{2} \frac{f_{-n}^{1,0}(1,\lambda_0\pi)}{f_{-n}(1,\lambda_0\pi)} \right) \label{pert}
\end{align}
by symmetry, since $|p_n|=|p_{-n}|$.

Let us simplify these expressions. The ratio before the infinite sum in \eqref{pert} evaluates to
\[
2 \frac{f_0^{1,0}(1,\lambda_0\pi)}{f_0^{0,1}(1,\lambda_0\pi)}\ = 2 \frac{M^\prime(a_0,1,z) \cdot \beta/\pi}{\frac{\partial M}{\partial a}(a_0,1,z) \cdot (-1/2\beta)} = c,
\]
by definition of $f_0$ and $c$ and remembering that $M(a_0,1,z)=0$ from \eqref{eq:gsmag}. We claim the denominator of $c$ is nonzero. Since $M(a_0,1,z)=0$, a parametric derivative formula due to Son \cite[Chapter 4]{S14} simplifies to tell us that
\begin{align*}
\frac{\partial M}{\partial a} (a_0,1,z) = -U(a_0,1,z) \Gamma(a_0) \int_0^z e^{-x}M(a_0,1,x)^2 \, dx
\end{align*}
where $U$ is the second standard Kummer function. Clearly the last integral is positive, and $U(a_0,1,z) \ne 0$ when $M(a_0,1,z)=0$, as explained in \cite[Chapter 4]{S14}. Thus the denominator of $c$ is nonzero.   

Next, write $q_n$ for the factor $(\cdots)$ in \eqref{pert}. By substituting the definitions of $f_n$ and $f_{-n}$ into \eqref{pert} we see that 
\begin{align}
q_n = 1+n -z + z ( \log M )^\prime(a_0,n+1,z) + z ( \log M )^\prime(a_0+n,n+1,z) , \label{eq:qndef}
\end{align}
where the derivatives are evaluated at the specific value $z = \beta/2\pi$ and we recall the definition $a_0=(1-\lambda_0 \pi/\beta)/2$. Thus $q_n$ has the form claimed in the theorem. 

We will show the first coefficient vanishes: $q_1=0$. Start with the definition
\[
q_1 = 2 -z + z ( \log M )^\prime(a_0,2,z) + z ( \log M )^\prime(a_0+1,2,z) .
\] 
First observe that $a_0<0$, since $\lambda_0>\beta/\pi$ (either by domain monotonicity and comparison with the first Landau level on the plane, or else by properties of Kummer functions \cite{S14}). 
Then note that 
\[
M(a_0+1,2,\cdot) = \frac{1}{a_0} M^\prime(a_0,1,\cdot)
\]
by the identity \cite[(13.3.15)]{NIST}. Further, 
\[
M(a_0,2,\cdot) = \frac{1}{a_0-1} \big( M^\prime(a_0,1,\cdot) - M(a_0,1,\cdot) \big)
\]
by expanding \cite[(13.3.20)]{NIST}. Substituting these last two formulas into the expression above for $q_1$ shows that
\[
q_1 = 2 -z + z \frac{M^{\prime \prime}(a_0,1,z) - M^\prime(a_0,1,z)}{M^\prime(a_0,1,z) - M(a_0,1,z)} + z \frac{M^{\prime \prime}(a_0,1,z)}{M^\prime(a_0,1,z)} .
\]
Since $M(a_0,1,z)=0$ by \eqref{eq:gsmag}, we may simplify to obtain
\[
q_1 = 2 - 2z + 2 \frac{zM^{\prime \prime}(a_0,1,z)}{M^\prime(a_0,1,z)} .
\]
The Kummer differential equation \eqref{eq:kummerDE} lets us substitute for $M^{\prime \prime}$, leading to
\[
q_1 = 2 -2z + 2 \frac{(z-1)M^\prime(a_0,1,z)}{M^\prime(a_0,1,z)} = 0
\]
as we wanted to show.

Finally we show $q_n=n+O(1)$. First, for any fixed $a,b,\zeta \in \R$ we have
\[
M(a,n+b,\zeta) \to 1 \qquad \text{and} \qquad M^\prime(a,n+b,\zeta) \to 0
\]
as $n \to \infty$, by the series definition of the Kummer function \cite[13.2.2]{NIST}. Second, 
\[
M(n+a,n+b,\zeta) \to e^\zeta \qquad \text{and} \qquad M^\prime(n+a,n+b,\zeta) \to e^\zeta
\]
as $n \to \infty$, again by using the series for the Kummer function. Hence the definition \eqref{eq:qndef} implies that $q_n=n+1+o(1)$ as $n \to \infty$.

\section{\bf Pauli eigenvalues}
\label{Peigen_proof}

\subsection*{Ellipticity of the numerator} The numerator of the Pauli-Rayleigh quotient \eqref{RayPauli} decouples as follows.
\begin{lemma}[Decoupling of numerator] \label{le:raycalc}
For $\psi= \left( \begin{smallmatrix} \psiplus \\ \psiminus \end{smallmatrix} \right)$ belonging to $H^1_0(\Omega;\C^2)$ we have
\begin{align*}
& \int_\Omega \left| \sum_{m=1}^2 \sigma_m (i\partial_m + F_m) \psi \right|^2 \, dx \\
& = \int_\Omega \Big( |(i\nabla + F) \psiplus |^2 - \frac{\beta}{A} |\psiplus|^2 \Big) dx +  \int_\Omega  \Big( |(i\nabla + F) \psiminus |^2 + \frac{\beta}{A} |\psiminus|^2 \Big) dx .
\end{align*}
\end{lemma}
\begin{proof}
We have by direct calculation (using the definition of the Pauli matrices) that
\begin{align}
& \left| \sum_{m=1}^2 \sigma_m (i\partial_m + F_m) \psi \right|^2 \notag \\
& = \big| (i\partial_1 + F_1)\psiplus + i(i\partial_2 + F_2) \psiplus \big|^2 + \big| (i\partial_1 + F_1)\psiminus - i(i\partial_2 + F_2) \psiminus \big|^2 . \label{eq:raycalc1}
\end{align}
Expand the squares to obtain
\begin{align*}
\left| \sum_{m=1}^2 \sigma_m (i\partial_m + F_m) \psi \right|^2
& = |(i\nabla + F)\psiplus|^2 + 2 \Real (i\partial_1 + F_1)\psiplus \, \overline{i(i\partial_2 + F_2)\psiplus} \\
& + |(i\nabla + F)\psiminus|^2 - 2 \Real (i\partial_1 + F_1)\psiminus \, \overline{i(i\partial_2 + F_2)\psiminus} .
\end{align*}
Then integrate the first cross-term as follows. One has 
\begin{align*}
& 2 \Real \int_\Omega (i\partial_1 + F_1)\psiplus \, \overline{i(i\partial_2 + F_2)\psiplus} \, dx \\
& = 2 \Real \int_\Omega \Big( -i \partial_1 \psiplus \overline{\partial_2 \psiplus} + \partial_1 \psiplus F_2 \overline{\psiplus} - F_1 \psiplus \overline{\partial_2 \psiplus} - i F_1 F_2 |\psiplus|^2 \Big) dx .
\end{align*}
The fourth term is imaginary, and so can be discarded. The first term can be integrated by parts twice to obtain its negative, by using the Dirichlet boundary conditions, and hence the integral of the first term must equal zero. Thus we are left with the integral of the second and third terms, so that
\begin{align*}
2 \Real \int_\Omega (i\partial_1 + F_1)\psiplus \, \overline{i(i\partial_2 + F_2)\psiplus} \, dx
& = \int_\Omega \Big( F_2 \partial_1 |\psiplus|^2 - F_1 \partial_2 |\psiplus|^2 \Big) dx \\
& = \int_\Omega (- \partial_1 F_2 + \partial_2 F_1) |\psiplus|^2 \, dx \qquad \text{by parts} \\
& = - \frac{\beta}{A} \int_\Omega |\psiplus|^2 \, dx ,
\end{align*}
where once again the boundary terms have vanished in the integration by parts because $\psiplus \in H^1_0(\Omega;\C)$.  The analogous formula holds for $\psiminus$, and so the lemma follows. 
\end{proof}

\begin{lemma}[Ellipticity of the numerator] \label{le:rayelliptic}
For $\psi= \left( \begin{smallmatrix} \psiplus \\ \psiminus \end{smallmatrix} \right)$ belonging to $H^1_0(\Omega;\C^2)$,
\[
\int_\Omega \left| \sum_{m=1}^2 \sigma_m (i\partial_m + F_m) \psi \right|^2 \, dx \notag \\
\geq \frac{1}{2} \int_\Omega |\nabla \psi|^2 \, dx - C \int_\Omega |\psi|^2 \, dx 
\]
where $\nabla \psi= \left( \begin{smallmatrix} \nabla \psiplus \\ \nabla \psiminus \end{smallmatrix} \right)$ and the constant can be chosen as $C=\lVert F \rVert_{L^\infty(\Omega)}^2 + |\beta|/A$.
\end{lemma}
\begin{proof}
Combine \autoref{le:raycalc} with the elementary inequality $|a+b|^2 \geq \frac{1}{2}|a|^2 - |b|^2$. 
\end{proof}

\subsection*{Proof of \autoref{Peigen}}
Assume $\beta>0$. (The proof is similar when $\beta<0$.) The Rayleigh quotient for the shifted Pauli eigenvalue $\lambda^P_j+\beta/A$ equals
\begin{equation*} 
\Ray_\text{shifted}[\psi] = \frac{\int_\Omega |(i\nabla + F) \psiplus |^2 \, dx +  \int_\Omega  \Big( |(i\nabla + F) \psiminus |^2 + \frac{2\beta}{A} |\psiminus|^2 \Big) dx}{\int_\Omega |\psi|^2 \, dx} ,
\end{equation*}
as we see by adding $\beta/A$ to the definition \eqref{RayPauli} of the unshifted Rayleigh quotient and then substituting the expression for its numerator from \autoref{le:raycalc}. 

Now one may prove the theorem by adapting straightforwardly the proof of \autoref{Deigen}, using in the course of the proof that $\beta/G \leq \beta$.  The equality statement for the first eigenvalue follows immediately from the equality statement in \autoref{Deigen}, since the lowest Pauli eigenvalue is related to the lowest eigenvalue of the magnetic Laplacian by $\lambda^P_1+|\beta|/A=\lambda_1$. 

\subsection*{Complex form of the Rayleigh quotient}
By substituting $F_1=-x_2\beta/2A$ and $F_2=x_1\beta/2A$ into \eqref{eq:raycalc1} we obtain the complex form of the Rayleigh quotient, which remains valid no matter what boundary conditions $\psi$ might satisfy:
\[
\Ray_P[\psi] = 4 \, \frac{\int_\Omega \big| (\overline{\partial} + \beta z/4A ) \psiplus \big|^2 + \big| (\partial - \beta \overline{z}/4A ) \psiminus \big|^2 \, dx}{\int_\Omega |\psi|^2 \, dx}
\]
where $\partial=\partial/\partial z$ is the complex derivative. Hence the Rayleigh quotient vanishes if and only if %
\begin{equation} \label{eq:rayz}
\psiplus = e^{-\beta |z|^2/4A} f_+(z) \qquad  \text{and} \qquad \psiminus = e^{\beta |z|^2/4A} \overline{f_-(z)} 
\end{equation}
for some holomorphic functions $f_+$ and $f_-$. We deduce that the zero modes form an infinite dimensional subspace of $H^1$, and so the numerator of the Rayleigh quotient is definitely not elliptic on $H^1$. To learn about zero modes on the whole plane, readers can consult the Aharanov--Casher theorem \cite{AC79}. 

\subsection*{Positivity of the first Pauli--Dirichlet eigenvalue}
The first Dirichlet eigenvalue is nonnegative, since the Rayleigh quotient is nonnegative. If the first eigenvalue were zero then the Rayleigh quotient of the first eigenfunction $\psi$ would equal zero, implying \eqref{eq:rayz}. The holomorphic functions $f_+$ and $f_-$  would then be forced to vanish identically, by the Dirichlet boundary condition, and so $\psi \equiv 0$, which is impossible. Hence the first eigenvalue must be positive. 

\subsection*{Splitting of the spectrum, and an alternative proof of \autoref{Peigen}}
The decoupling of the Rayleigh quotient in \autoref{le:raycalc} implies a decoupling of the eigenvalue equations into separate equations for each component of the spinor:
\begin{align}
(i\nabla + F)^2 \psiplus - (\beta/A) \psiplus & = \lambda \psiplus \label{eq:de1} \\ 
(i\nabla + F)^2 \psiminus + (\beta/A) \psiminus & = \lambda \psiminus \label{eq:de2}
\end{align}
Hence the introduction of spin into the quantum system splits the spectrum of the Dirichlet magnetic Laplacian into two copies, with one copy shifted up by $\beta/A$ and another shifted down by the same amount. More precisely, if we write $H_{\text{mag}} = (i\nabla + F)^2$ for the magnetic Laplacian two dimensions, then the Dirichlet spectrum of the Pauli operator is 
\[
\spec(H_P) = \left[ \spec(H_{\text{mag}}) - \frac{\beta}{A} \right] \cup \left[ \spec(H_{\text{mag}}) + \frac{\beta}{A} \right] ,
\]
with multiplicities being respected by the union. 

One obtains this same result at the level of operators, of course: first expand the definition of the Pauli operator to show that $H_P = (i\nabla + F)^2 I - \sigma \cdot B$,  and then use that the magnetic field is vertical to find $\sigma \cdot B=(\beta/A) \sigma_3$, which gives \eqref{eq:de1}-\eqref{eq:de2}.

We could have proved \autoref{Peigen} by using this splitting of the spectrum, as we now explain. Shifting the spectrum up by $|\beta|/A$ and multiplying by $A/G$ to obtain a scale invariant expression gives 
\[
\left( \spec(H_P) + \frac{|\beta|}{A} \right) \frac{A}{G} = \left[ \spec(H_{\text{mag}})\frac{A}{G} \right] \cup \left[ \spec(H_{\text{mag}})\frac{A}{G} + \frac{2|\beta|}{G} \right] .
\]
Note that the first spectrum on the right is not shifted, and the second is shifted by an amount $2|\beta|/G$ that is maximal for the disk (since $G$ is minimal for the disk). Hence one can prove \autoref{Peigen} by starting with the magnetic Laplacian result \autoref{Deigen} and extending the Hardy--Littlewood-P\'{o}lya  majorization technique to handle the union of a sequence and a shifted copy of the same sequence. We omit these proofs.

We chose to follow a more direct approach to proving \autoref{Peigen}, in the hope that it might help some future researcher to treat  non-Dirichlet boundary conditions.

\section*{Acknowledgments}
This work was partially supported by a grant from the Simons Foundation (\#204296 to Richard Laugesen), the National Science Foundation grant DMS-0803120 (Laugesen), and the Polish National Science Centre (NCN) grant 2012/07/B/ST1/03356 (Siudeja).

We are grateful to the Banff International Research Station for funding our participation in the workshop on ``Spectral Theory of Laplace and Schr\"{o}dinger Operators" (July/August 2013), during which this paper was almost finalized.

\newcommand{\doi}[1]{%
 \href{http://dx.doi.org/#1}{doi:#1}}
\newcommand{\arxiv}[1]{%
 \href{http://front.math.ucdavis.edu/#1}{ArXiv:#1}}
\newcommand{\mref}[1]{%
\href{http://www.ams.org/mathscinet-getitem?mr=#1}{#1}}

\end{document}